\documentclass[a4paper]{amsart}\usepackage[a4paper,margin=2.5cm,includefoot]{geometry}
\usepackage[utf8]{inputenc}
\usepackage[english]{babel}
\usepackage[T1]{fontenc}

\usepackage{amsmath,amssymb,mathtools,bm}
\usepackage{csquotes}
\usepackage{amstext}
\usepackage{amsthm}
\usepackage{bbm}
\usepackage{enumerate}
\usepackage[svgnames]{xcolor}
\usepackage[colorlinks,allcolors=MidnightBlue,citecolor=DarkRed]{hyperref}
\usepackage[nameinlink,capitalise]{cleveref}
\usepackage{braket}
\usepackage{dsfont}
\usepackage{color}
\usepackage{tikz}
\usepackage{pgfplots}

\makeatletter
\def\@seccntformat#1{%
	\protect\textup{\protect\@secnumfont
		\ifnum\pdfstrcmp{subsection}{#1}=0 \bfseries\fi
		\ifnum\pdfstrcmp{subsubsection}{#1}=0 \itshape\fi
		\csname the#1\endcsname
		\protect\@secnumpunct
	}%
}
\renewcommand{\@upn}{}
\DeclareRobustCommand{\crefnosort}[1]{%
	\begingroup\@cref@sortfalse\cref{#1}\endgroup
}
\makeatother

\makeatletter
\AddToHook{cmd/appendix/before}{\def\cref@section@alias{appendix}}
\makeatother

\numberwithin{equation}{section}
\newtheorem{thm}{Theorem}[section]
\newtheorem{lem}[thm]{Lemma}
\AddToHook{env/lem/begin}{\crefalias{thm}{lem}}
\newtheorem{prop}[thm]{Proposition}
\AddToHook{env/prop/begin}{\crefalias{thm}{prop}}
\newtheorem{cor}[thm]{Corollary}
\AddToHook{env/cor/begin}{\crefalias{thm}{cor}}

\theoremstyle{definition}

\AddToHook{env/defn/begin}{\crefalias{thm}{defn}}

\AddToHook{env/hyp/begin}{\crefalias{thm}{hyp}}
\renewcommand*{\thehyp}{\Alph{hyp}}

\theoremstyle{remark}
\newtheorem{rem}[thm]{Remark}
\AddToHook{env/rem/begin}{\crefalias{thm}{rem}}

\AddToHook{env/ex/begin}{\crefalias{thm}{ex}}

\crefname{hyp}{Hypothesis}{Hypotheses}\Crefname{hyp}{Hypothesis}{Hypotheses}
\crefname{lem}{Lemma}{Lemmas}\Crefname{lem}{Lemma}{Lemmas}
\crefname{thm}{Theorem}{Theorems}\Crefname{thm}{Theorem}{Theorems}
\crefname{prop}{Proposition}{Propositions}\Crefname{prop}{Proposition}{Propositions}
\crefname{enumi}{}{}\Crefname{enumi}{}{}
\creflabelformat{enumi}{#2(#1)#3}
\crefname{equation}{}{}\Crefname{equation}{}{}
\crefname{rem}{Remark}{Remarks}\Crefname{rem}{Remark}{Remarks}
\crefname{ex}{Example}{Examples}\Crefname{ex}{Example}{Examples}

\makeatletter
\renewcommand{\@upn}{} 
\makeatother

\usepackage[inline]{enumitem}
\newlist{enumthm}{enumerate}{1} 
\setlist[enumthm]{label=\upshape(\roman*),ref=\thethm\,(\roman*)}  
\crefalias{enumthmi}{thm} 
\newlist{enumcor}{enumerate}{1}
\setlist[enumcor]{label=\upshape(\roman*),ref=\thecor\,(\roman*)}
\crefalias{enumcori}{cor}
\newlist{enumlem}{enumerate}{1}
\setlist[enumlem]{label=\upshape(\roman*),ref=\thelem\,(\roman*)}
\crefalias{enumlemi}{lem}
\newlist{enumprop}{enumerate}{1}
\setlist[enumprop]{label=\upshape(\roman*),ref=\theprop\,(\roman*)}
\crefalias{enumpropi}{prop}
\newlist{enumhyp}{enumerate}{1}
\setlist[enumhyp]{label=\upshape(\roman*),ref=\thehyp\,(\roman*)}
\crefalias{enumhypi}{hyp}
\newlist{enumproof}{enumerate*}{1}
\setlist[enumproof]{label=\upshape(\roman*)}
\newlist{enumdef}{enumerate}{1}
\setlist[enumdef]{label=\upshape(\roman*),ref=\thedefn\,(\roman*)}
\crefalias{enumdefi}{defn}
\newlist{enumrem}{enumerate}{1}
\setlist[enumrem]{label=\upshape(\roman*),ref=\therem\,(\roman*)}
\crefalias{enumremi}{rem}

\makeatletter
\newcounter{subcreftmpcnt} %
\newcommand\romansubformat[1]{(\roman{#1})} 
\def\subcref{\@ifstar\@@subcref\@subcref}
\newcommand\@subcref[2][\romansubformat]{%
	\ifcsname r@#2@cref\endcsname
	\cref@getcounter {#2}{\mylabel}%
	\setcounter{subcreftmpcnt}{\mylabel}%
	\hyperref[#2]{\romansubformat{subcreftmpcnt}}%
	\else ?? \fi}   
\newcommand\@@subcref[2][\romansubformat]{%
	\ifcsname r@#2@cref\endcsname
	\cref@getcounter {#2}{\mylabel}%
	\setcounter{subcreftmpcnt}{\mylabel}%
	\romansubformat{subcreftmpcnt}%
	\else ?? \fi}   
\makeatother

\makeatletter
\DeclareRobustCommand{\crefnosort}[1]{%
	\begingroup\@cref@sortfalse\cref{#1}\endgroup
}
\makeatother

\def\endstepsymbol{$\lozenge$}
\def\endclaimsymbol{$\lozenge$}
\newcounter{proofstep}
\AtBeginEnvironment{proof}{\setcounter{proofstep}{0}}

\crefname{proofstep}{Step}{Steps}
\Crefname{proofstep}{Step}{Steps}
\newcounter{proofclaim}
\AtBeginEnvironment{proof}{\setcounter{proofclaim}{0}}

\crefname{proofclaim}{Claim}{Claims}
\Crefname{proofclaim}{Claim}{Claims}

\newcommand{\cB}{{\mathcal B}}
\newcommand{\cE}{{\mathcal E}}\newcommand{\cF}{{\mathcal F}}

\newcommand{\cM}{{\mathcal M}}
\newcommand{\cQ}{{\mathcal Q}}

\newcommand{\fB}{{\mathfrak B}}

\newcommand{\fV}{{\mathfrak V}}

\newcommand{\fg}{{\mathfrak g}}

\newcommand{\fn}{{\mathfrak n}}
\newcommand{\fq}{{\mathfrak q}}

\newcommand{\BC}{{\mathbb C}}

\newcommand{\BN}{{\mathbb N}}
\newcommand{\BR}{{\mathbb R}}

\usepackage{dsfont} 

\newcommand{\DSE}{{\mathds E}}

\newcommand{\DSP}{{\mathds P}}

\newcommand{\dsone}{{\mathds 1}}

\usepackage{mathrsfs}

\newcommand{\sD}{{\mathscr D}}

\newcommand{\sQ}{{\mathscr Q}}

\newcommand{\sW}{{\mathscr W}}

\newcommand{\sfB}{{\mathsf B}}

\newcommand{\sfS}{{\mathsf S}}\newcommand{\sfU}{{\mathsf U}}
\newcommand{\sfV}{{\mathsf V}}

\newcommand{\sfc}{{\mathsf c}}
\newcommand{\sfd}{{\mathsf d}}\newcommand{\sfe}{{\mathsf e}}\newcommand{\sff}{{\mathsf f}}
\newcommand{\sfg}{{\mathsf g}}
\newcommand{\sfl}{{\mathsf l}}
\newcommand{\sfm}{{\mathsf m}}\newcommand{\sfn}{{\mathsf n}}\newcommand{\sfo}{{\mathsf o}}
\newcommand{\sfr}{{\mathsf r}}
\newcommand{\sfs}{{\mathsf s}}
\newcommand{\sfx}{{\mathsf x}}
\newcommand{\sfy}{{\mathsf y}}\newcommand{\sfz}{{\mathsf z}}

\newcommand{\rme}{{\mathrm e}}

\newcommand{\IN}{\BN}\newcommand{\IR}{\BR}\newcommand{\IC}{\BC}

\newcommand{\PP}{\DSP}\newcommand{\EE}{\DSE}

\newcommand{\eps}{\varepsilon}\newcommand{\ph}{\varphi}

\newcommand{\e}{\rme}\newcommand{\Id}{\dsone} \renewcommand{\d}{\sfd}

\renewcommand{\Im}{\operatorname{Im}}

\DeclareFontFamily{U}{mathx}{\hyphenchar\font45}
\DeclareFontShape{U}{mathx}{m}{n}{
	<5> <6> <7> <8> <9> <10>
	<10.95> <12> <14.4> <17.28> <20.74> <24.88>
	mathx10
}{}
\DeclareSymbolFont{mathx}{U}{mathx}{m}{n}
\DeclareFontSubstitution{U}{mathx}{m}{n}
\DeclareMathAccent{\widecheck}{0}{mathx}{"71}
\DeclareMathAccent{\wideparen}{0}{mathx}{"75}

\newcommand{\chr}{\mathbf 1}

\newcommand{\norm}[1]{\lVert#1\lVert}

\newcommand{\FGamma}{\Gamma}
\newcommand{\FS}{\cF}\newcommand{\dG}{\sfd\FGamma}\newcommand{\ad}{a^\dagger}

\title[Feynman--Kac Formula for the Renormalized Spin Boson Model]{A Feynman--Kac Formula for the Subcritical\\ Ultraviolet-Renormalized Spin Boson Model}
\author{Daniel M. Fr\"ohlich}
\author{Benjamin Hinrichs}
\address{Universit\"at Paderborn, Institut f\"ur Mathematik, Institut f\"ur Photonische Quantensysteme, Warburger Str. 100, 33098 Paderborn, Germany}
\email{danif1@upb.de, benjamin.hinrichs@math.upb.de}

\usepackage{xcolor}\usepackage{cancel}
\usepackage{mathtools}

\DeclareMathOperator{\Span}{span}

\newcommand{\HSB}{H_{\sfS\sfB}}

\begin{document}

\begin{abstract} 
	\noindent
	We prove a Feynman--Kac formula (FKF) for the self-energy renormalized spin boson Hamiltonian,
    describing a two-state quantum system linearly coupled to a bosonic quantum field.
    Similar to recent FKFs for the Fr\"ohlich polaron and the non- and semi-relativistic Nelson models,
    it yields a probabilistic treatment of the spin as a jump process, but treats the field on the usual bosonic Fock space.
	As an application, we prove that the existence of ground states for infrared-regular models persists the removal of an ultraviolet cutoff.
\end{abstract}

\maketitle

\section{Introduction}

The spin boson model is one of the simplest non-trivial models of field-matter interactions,
in which a two-state system (`spin') is linearly coupled to a bosonic quantum field.
Despite its simplicity, it has found its way into various physical applications,
ranging from quantum optics via simple open quantum systems to caricature models of quantum electrodynamics.
Since the references are numerous, we here defer the reader to some reviews given in \cite{Leggettetal.1987,Spohn.1989}
or more recently in \cite[\S~3]{FalconiHinrichsValentinMartin.2026}.

Albeit both the large coupling and large boson momentum regime
have already been of longstanding interest for the spin-boson model
\cite{Leggettetal.1987}, they have especially received recent attention
in the mathematical physics literature. Especially, as one important contribution,
Dam and M\o{}ller \cite{DamMoller.2018b} proved no-go theorems for supercritical
interactions on the Hilbert space of the free model and further gave an account
for both the infrared and the ultraviolet problem. A corresponding wavefunction
renormalization scheme, addressing both singularities simultaneously, has recently
been developed in \cite{FalconiHinrichsValentinMartin.2026}.

\smallskip
In this article, we focus on spin-boson models with up-to-critical interactions,
for which a so-called self-energy ultraviolet renormalization scheme converges.
Explicitly, if $H_\Lambda$ denotes the regularized Hamiltonian with an ultraviolet cutoff
$\Lambda\in[0,\infty)$, then the renormalized Hamiltonian is obtained as the limit $H_\infty=\lim_{\Lambda\to\infty}(H_\Lambda-E_\Lambda)$
for an appropriate self-energy term $E_\Lambda\in\IR$. Such schemes (and the failure to converge for supercritical interactions)
were already studied in \cite{DamMoller.2018b} and explicit constructions of the renormalized Hamiltonian
$H_\infty$ were recently given in \cite{HinrichsLampartValentinMartin.2024,AlvarezLillLonigroValentinMartin.2025}.

Self-energy renormalization schemes possess a long standing history, going back to the Nelson model introduced in \cite{Nelson.1964},
later revisited by many authors; see, e.g., \cite{Ammari.2000,GriesemerWuensch.2018,MatteMoller.2018,LampartSchmidt.2019} and references therein.
One way to construct the renormalized operator hereby uses functional integration methods introduced by means of a Feynman--Kac formula.
This has been first achieved using a Euclidean field approach in \cite{GubinelliHiroshimaLorinczi.2014}
and later using a Fock space-valued Feynman--Kac formula (FKF), as introduced in \cite{GueneysuMatteMoller.2017},
in \cite{MatteMoller.2018}. Other models as the original Nelson model, e.g., a relativistic version and the Fr\"ohlich polaron model,
were studied using this approach in \cite{HinrichsMatte.2022,HinrichsMatte.2023} and \cite{HinrichsMatte.2024}, respectively.

We add to this picture, by proving a Fock space-valued FKF  for the self-energy renormalized spin boson Hamiltonian $H_\infty$.
Explicitly, as usual denoting by $H_\sff$ the free field Hamiltonian and by $a$ and $\ad$ the boson annihilation and creation operators on the bosonic Fock space $\FS$, respectively, our formula takes the form
\begin{align}\label{eq:FKF} e^{-tH_\infty}\psi (x) = \EE^x\Bigg[e^{u_t}e^{\ad(-U_t^+)}e^{-tH_\sff}e^{a(-U_t^-)}\psi(X_t)\Bigg], \quad x=\pm 1,\ \psi\in \IC^2\otimes\FS\cong L^2(\{\pm1\};\FS), \end{align}
where $\EE^x$ is the expectation with respect to a probability measure on the right-continuous spin paths $X:[0,\infty)\to\{-1,+1\}$ with support on the paths satisfying $X_0=x$,
$u_t$ is a real-valued stochastic process and $U_t^\pm$ are stochastic processes taking values in the one-boson Hilbert space.

\smallskip
Denoting by $\Omega_\downarrow$ the ground state of the decoupled model, the FKF \cref{eq:FKF} implies
\[ \braket{\Omega_\downarrow,e^{-tH_\infty}\Omega_\downarrow} = \EE^x[e^{u_t}], \]
a relation which allows to study the ground state regime of $H_\infty$, by solely analyzing the process $u_t$.
For the ultraviolet-regularized model $H_\Lambda$ with $\Lambda<\infty$, $u_t$ takes the well-known form
\begin{align}\label{eq:utreg}  u_t = \int_0^t\int_0^t W_\Lambda(s-r)X_sX_r\sfd s\sfd r - t E_\Lambda, \end{align}
which has been widely used to study the ground state energy \cite{Abdessalam.2011} as well as the existence of a ground state
for the infrared-regular \cite{HirokawaHiroshimaLorinczi.2014} and the infrared-singular model \cite{HaslerHinrichsSiebert.2021b,HaslerHinrichsSiebert.2021c,BetzHinrichsKraftPolzer.2025,HinrichsPolzer.2025},
by employing the observation that this provides a connection between the spin boson model and the
one-dimensional Ising model \cite{EmeryLuther.1974,FannesNachtergaele.1988,SpohnDuemcke.1985,Spohn.1989}.

We explicitly construct the process $u_t$ without an ultraviolet cutoff and show that for infrared-regular interactions
a unique ground state exists independent of the presence of the UV cutoff. In this spirit, we provide a first justification
for the heuristic folklore that existence of a ground state is an infrared property.

To further investigate this fact,
one would need to study the ground state behavior of the infrared singular model as well. In this case, in presence of an ultraviolet cutoff,
it is known that ground states exist at small coupling, but cease exist at large coupling \cite{Spohn.1989,HaslerHinrichsSiebert.2021a,BetzHinrichsKraftPolzer.2025}.
This is a subtle property corresponding to the phase transition of the long-range one-dimensional Ising model \cite{Dyson.1969,FrohlichSpencer.1982,ImbrieNewman.1988}. 
Since the proof of this phase transition for the spin boson model heavily relies on this correspondence, we strongly believe that
our FKF is an important ingredient to this open problem for the model without an ultraviolet cutoff, but leave this interesting
question for future work. Especially, we emphasize that a complication in this endeavour is due to the fact that our construction of $u_t$
does not take the form \cref{eq:utreg} anymore, whence the interpretation as a one-dimensional Ising model does not directly carry over to the renormalized case.

\subsection*{Structure of the Article}
In \cref{sec:SB}, we give a precise construction of the Hamiltonian $H_\infty$ and state the two main results:
convergence of a self-energy renormalization scheme (\cref{thm:SEren}) and existence of a unique ground state (\cref{thm:GS}).
In \cref{sec:FK}, we introduce the probabilistic ingredients of our approach, i.e., the processes $u_t$ and $U_t^\pm$
and prove the FKF (\cref{thm:FKF}). In \cref{sec:app}, we then study the ground state regime of $H_\infty$, by proving ergodicity (\cref{prop:erg}) and the existence of ground states \cref{thm:GS}.

\subsection*{Acknowledgments}
A FKF of the type \cref{eq:FKF} in presence of an ultraviolet cutoff was already derived in DMF's master thesis, which was supervised by the second author and written at Leipzig University.
In this course, the authors want to thank Marcel Schmidt for introducing us, for refereeing on the thesis and for valuable suggestions which entered this manuscript.
The authors acknowledge support by the Ministry of Culture and Science of the State of North Rhine-Westphalia within the project `PhoQC' (Grant Nr. PROFILNRW-2020-067).
BH was supported by the German Research Foundation (DFG) within the scientific network `A(E)MP -- Appearance of the Effective Mass in Polaron Models´ (Grant No. 569490025). 

\section{The Self-Energy Renormalized Spin Boson Model}
\label{sec:SB}
In this section, we give a precise construction of the self-energy renormalized spin boson Hamiltonian,
starting with a short introduction on Fock space notation (\cref{subsec:Fock}) and
the standard ultraviolet regularized spin boson model (\cref{subsec:regSB}).
We then construct the self-energy renormalized model (\cref{subsec:Hren})
and state our result on the existence of ground states (\cref{subsec:GS}).

\subsection{Fock Space Calculus}
\label{subsec:Fock}

We start, by introducing the Hilbert space of the bosonic field -- the so-called {\em Fock space} $\FS$.
Throughout this paper, we assume the momentum space for the bosons $(\cM,\mu)$ to be a $\sigma$-finite measure space
and thus the single-boson Hilbert space is $L^2(\cM,\mu)$. The bosonic Fock space is then defined as 
\begin{align}
    \label{eq:Fock}
    \FS \coloneqq \IC\oplus\bigoplus_{n=1}^\infty L^2_{\sfs\sfy\sfm}(\cM^n,\mu^{\otimes n}),
\end{align}
where the index $\sfs\sfy\sfm$ indicates symmetry with respect to permutations  in the $n$ variables of the $n$-boson subspace $L^2_{\sfs\sfy\sfm}(\cM^n,\mu^{\otimes n})$.

For our purposes, an important set of test vectors are the so-called {\em exponential vectors} (if normalized also called {\em coherent states})
\begin{align}
    \label{eq:exp}
    \eps(f) \coloneqq 1\oplus\bigoplus_{n=1}^\infty \frac{1}{\sqrt{n!}}f^{\otimes n}, \qquad f\in L^2(\cM,\mu)
\end{align}
where as usual the tensor product of $g,h\in L^2(\cM,\mu)$ is $g\otimes h:(k,r)\mapsto g(k)h(r)\in L^2(\cM^2,\mu^{\otimes 2})$.
Given $D\subset L^2(\cM,\mu)$, we will further write
\begin{align}
    \label{eq:ED}
    \cE(D)\coloneqq \Span \{\eps(f)|f\in D\}.
\end{align}
Let us now introduce the operators on $\FS$ relevant for this article, which we here for convenience define by their action on $\cE(D)$. For more extensive introductions, we refer to \cite{Parthasarathy.1992,Arai.2018}.

Given $f\in L^2(\cM,\mu)$ and a unitary operator $B$ on $L^2(\cM,\mu)$, we define the {\em Weyl operator} $\sW(f)$ and the {\em second quantization operator} $\Gamma(U)$ by
\begin{align}
    \label{eq:Weyl}
    \sW(f)\eps(g) \coloneqq e^{-\frac12\norm{f}^2 - \braket{f,g}}\eps(f+g),\quad \Gamma(U)\eps(f)=\eps(Uf),
\end{align}
which both uniquely extend to unitary operators on $\FS$ and satisfy the relations
\begin{align}
    \sW(f)\sW(g)=e^{-i\Im\braket{f,g}}\sW(f+g) \quad\mbox{and}\quad \Gamma(U)\Gamma(T)=\Gamma(UT).
\end{align}
Further, if $m$ is a selfadjoint operator on $L^2(\cM,\mu)$, then $t\mapsto \Gamma(e^{-itm})$ is a strongly continuous unitary group,
and we denote its selfadjoint generator by $\dG(m)$.
More explicitly, if $m:\cM\to\IR$ is a multiplication operator, then it acts on the $n$-boson subspace $L^2_{\sfs\sfy\sfm}(\cM^n,\mu^{\otimes n})$
as multiplication with $(k_1,\ldots,k_n)\mapsto \sum_{i=1}^n m(k_i)$. Especially, if $m\ge 0$, then also $\dG(m)\ge 0$.

Also for $f\in L^2(\cM,\mu)$, we define the {\em annihilation operator} $a(f)$ as the closure of the operator defined on $\cE(L^2(\cM,\mu))$ by
\begin{align}
    \label{def:ann}
    a(f)\eps(g) \coloneqq \braket{f,g}\eps(g).
\end{align}
Its adjoint $\ad(f)\coloneqq a(f)^*$ is the {\em creation operator} and they satisfy the {\em canonical commutation relations}
\begin{align}
    [a(f),a(g)] = [\ad(f),\ad(g)] = 0,\quad [a(f),\ad(g)]=\braket{f,g},
\end{align}
e.g.,
on $\cE(L^2(\cM,\mu))$.
Finally, we define the {\em field operator}
\begin{align}
    \label{eq:field}
    \ph(f) \coloneqq \overline{a(f)+\ad(f)},
\end{align}
where $\overline{\,\cdot\,}$ as usual denotes the operator closure. It is selfadjoint and satisfies
\begin{align}
    \label{eq:Weylfield}
    \sW(-if) = e^{-i\ph(f)}.
\end{align}
Furthermore, if $m:\cM\to[0,\infty)$ and $f\in L^2(\cM,\mu)$ satisfies $f/\sqrt m\in L^2(\cM,\mu)$, then we have the relative bound
\begin{align}
    \label{eq:simplerel}
    \norm{\ph(f)\psi} \le \eps\norm{f/\sqrt m}\norm{\dG(m)\psi} + \big(\tfrac4\eps\norm{f/\sqrt m} + \norm{f}\big)\norm{\psi},
    \quad\mbox{for all}\ \psi\in\sD(\ph(f))\subset \sD(\dG(m)),\ \eps>0.
\end{align}
The Fock-space part of the Feynman--Kac formula, which is our main result, will encompass terms of the form $e^{-t\dG(m)}e^{a(-f)}$, which yields a densely defined operator with bounded closure. Let us here again define this closure, by its action on exponential vectors
\begin{align}  
    \label{eq:It}
    I_t(m,f)\eps(g)\coloneqq e^{-\braket{f,g}}\eps({e^{-tm}g})
\end{align}
and with adjoint acting as (\cite[Rem.~5.2]{MatteMoller.2018})
\begin{align}
    \label{eq:Itstar}
    I_t(m,f)^*\eps(g)  = \eps(e^{-tm}g-f).
\end{align}
The so-defined operators are bounded for $f\in\sD(m^{-1/2})$, as is proven in \cite{GueneysuMatteMoller.2017} by expanding the exponentials in series.
We here also recall that there exists a universal $C>0$ such that
\begin{align}
    \label{eq:Ftbound}
    \norm{I_t(m,f)} \le C e^{8\|f/\min(1,\sqrt{tm})\|^2}.
\end{align}
\subsection{Regular Spin Boson Model}
\label{subsec:regSB}

We now introduce the spin boson model, starting with its usual Hamiltonian representation on $\IC^2\otimes\FS$.

Therefore, assume that $\omega:\cM\to[0,\infty)$ is the {\em boson dispersion relation} satisfying $\omega>0$ $\mu$-a.e.
Then, for all interactions
\begin{align}
    \label{eq:Vreg}
    v \in \fV_{\sfr\sfe\sfg} \coloneqq \big\{ f\in L^2(\cM,\mu)  \big| f/\sqrt\omega\in L^2(\cM,\mu) \big\},
\end{align}
we can now define the {\em regular spin boson Hamiltonian}
\begin{align}
	\HSB(v) \coloneqq \sigma_\sfz\otimes \Id_\FS + \Id_{\IC^2}\otimes\dG(\omega) + \sigma_\sfx\otimes \ph(v) + \big(1+\norm{\omega^{-1/2}v}^2\big)\Id_{\IC^2\otimes\FS} \qquad \mbox{on}\ \IC^2\otimes \FS.
\end{align}
It is manifestly selfadjoint and lower-bounded, by the Kato--Rellich theorem and the relative bound \cref{eq:simplerel}.
Note that the last term incorporates a self-energy shift convenient for the following presentation, especially when moving to the renormalized model for more singular interactions.

We will now first pass to a unitarily equivalent representation of the spin boson model, in which the interaction is diagonal on the spin space and the free energy of the spin is modeled as the graph Laplacian on $\{\pm 1\}$.
We will thus, for convenience in the probabilistic representation, map elements of $\IC^2\otimes\FS$ to $\FS$-valued functions on $\{\pm1\}$.
\begin{lem}
	\label{lem:unitary}
	We define the unitary $U:\IC^2\otimes \FS \to L^2(\{\pm 1\};\FS)$ by
	\begin{align}
		\begin{pmatrix}a\\b\end{pmatrix}\otimes \psi \mapsto \begin{cases} +1\mapsto  \frac1{\sqrt 2}(b+a)\psi\\-1\mapsto\frac1{\sqrt 2}(b-a)\psi.\end{cases}
	\end{align}
	Then, for all $v\in\fV_{\sfr\sfe\sfg}$, we have
    \begin{align}
        \label{eq:rewriting}
        \big(U\HSB (v) U^*\psi\big)(x) =  \big(\psi(x)-\psi(-x)\big) + \dG(\omega)\psi(x) + x\ph(v)\psi(x) + \norm{\omega^{-1/2}v}^2.
    \end{align}
\end{lem}
\begin{proof}
	The claim follows, by observing that \[\big(U(\sigma_ \sfx\otimes\Id_\FS)U^*\psi\big)(x) = x\psi(x)\]
    and that \[\big(U(\Id_{\IC^2\otimes\FS}+\sigma_\sfz\otimes\Id_\FS)U^*\psi\big)(x) = \psi(x)-\psi(-x),\]
    which in turn is immediate by direct calculation.
\end{proof}

\subsection{UV-Renormalized Spin Boson Model}
\label{subsec:Hren}
We will now use the right hand side of \cref{eq:rewriting} as definition for the Hamiltonian under consideration, which is the renormalized version of the spin boson Hamiltonian defined for interactions from $\fV_{\sfr\sfe\sfg}+\fV_{\sfU\sfV}$ with
the singular interaction space
\begin{align}
    \label{eq:VUV}
    \fV_{\sfU\sfV} \coloneqq \big\{ f:\cM\to\mu\ \mbox{measurable} \big| f/\omega\in L^2(\cM,\mu) \big\}.
\end{align}
The key observation in the definition of the renormalized spin boson model is that
the van Hove Hamiltonian \cite{vanHove.1952} $\dG(\omega) + x\ph(v) + \norm{\omega^{-1/2}v}^2$ with $v\in \fV_{\sfr\sfe\sfg}$ can
be self-energy renormalized as the selfadjoint lower-bounded Hamiltonian \cite{Derezinski.2003}
\begin{align}
    \label{eq:vHren}
    h(v,x) &= \sW(-x v_{\sfU\sfV}/\omega)\big(\dG(\omega)+x\ph(v_{\sfr\sfe\sfg}) + \norm{\omega^{-1/2}v_{\sfr\sfe\sfg}}^2\big)\sW(x v_{\sfU\sfV}/\omega), \qquad x=\pm1,\\
    \label{eq:decomp}
    &\mbox{where}
    \ v = v_{\sfr\sfe\sfg} + v_{\sfU\sfV},\ v_{\sfr\sfe\sfg}\in\fV_{\sfr\sfe\sfg},\ v_{\sfU\sfV}\in\fV_{\sfU\sfV},\ \braket{v_{\sfr\sfe\sfg},v_{\sfU\sfV}/\omega} = 0.
\end{align}
We stress that the right hand side of \cref{eq:vHren} is independent of the concrete decomposition of $v$ in \cref{eq:decomp},
due to the observation that $v_{\sfr\sfe\sfg} + v_{\sfU\sfV} = \tilde v_{\sfr\sfe\sfg} + \tilde v_{\sfU\sfV}$ implies $v_{\sfU\sfV}-\tilde v_{\sfU\sfV}\in\fV_{\sfr\sfe\sfg}$,
the Weyl relations \cref{eq:Weyl} and the transformation behavior of $\dG(\omega)$ and $\ph(v)$ under Weyl operators,
see for example \cite{Parthasarathy.1992,BratteliRobinson.1996} for details or \cite{HinrichsLampartValentinMartin.2024} for an explicit application to the spin boson model.
Furthermore, we stress that a decomposition of the type \cref{eq:decomp} always exists, e.g., by choosing $v_{\sfr\sfe\sfg}=v\chr_{\{\omega\le 1\}}$ and $v_{\sfU\sfV}=v\chr_{\{\omega> 1\}}$.
Finally, let us note that the construction directly implies that
\begin{align}
    \cE_{\sfr\sfe\sfn}(x,v) \coloneqq \sW(-x v_{\sfU\sfV})\cE(\sD(\omega)) = \cE(-x v_{\sfU\sfV} + \sD(\omega))
\end{align}
is a core for $h(v,x)$, once more noting the independence of the right hand side of the concrete decomposition of $v$ in \cref{eq:decomp}.

With this at hand, we can define the self-energy renormalized spin boson Hamiltonian for form factors $v\in \fV_{\sfr\sfe\sfg}+\fV_{\sfU\sfV}$ as
\begin{align}
    \label{eq:Hren}
    \begin{aligned}
        & \sD\big(H_{\sfr\sfe\sfn}(v)\big) = \big\{ \psi\in L^2(\{\pm1\},\FS)\big|\psi(x)\in \sD(h(v,x))=\sW(-x v_{\sfU\sfV}/\omega)\sD(\dG(\omega))\big\},
    \\&
    \big(H_{\sfr\sfe\sfn}(v) \psi\big)(x) \coloneqq \psi(x)-\psi(-x) + h(v,x)\psi(x).
    \end{aligned}
\end{align}
Its selfadjointness easily follows from that of $h(v,x)$
and a suitable core is given by
\begin{align}
    \label{eq:core}
    \bigoplus_{x=\pm 1}\cE_{\sfr\sfe\sfn}(x,v) = \big\{ \psi\in L^2(\{\pm1\},\FS) \big| \psi(x)\in \cE_{\sfr\sfe\sfn}(x,v),\ x=\pm 1 \big\}.
\end{align}
The following result is known, but also follows naturally from our Feynman--Kac representation of $H_{\sfr\sfe\sfn}(v)$.
\begin{thm}
    \label{thm:SEren}
    Let $(v_n)_{n\in\IN}\subset \fV_{\sfr\sfe\sfg}$ and $v\in \fV_{\sfr\sfe\sfg}+\fV_{\sfU\sfV}$ such that $v_n/(\sqrt{\omega}+\omega)\xrightarrow[L^2(\cM,\mu)]{n\to\infty} v/(\sqrt{\omega}+\omega)$.
    Then
    \begin{align*}
        U \HSB(v_n)U^* \xrightarrow{n\to\infty} H_{\sfr\sfe\sfn}(v) \quad\mbox{in the norm resolvent sense.}
    \end{align*}
\end{thm}
\begin{proof}
    Noting that $UH_{\sfS\sfB}(v)U^* = H_{\sfr\sfe\sfn}(v)$ for $v\in\fV_{\sfr\sfe\sfg}$, by \cref{lem:unitary},
    the claim follows from the fact that $v\mapsto e^{-t H_{\sfr\sfe\sfn}(v)}$ is continuous in operator norm w.r.t. the norm $\|v/(\sqrt{\omega}+\omega)\|$.
    This is proved in \cref{thm:FKF}.
\end{proof}
\begin{rem}
    Self-energy renormalizations of the spin boson model were derived in \cite{DamMoller.2018a,HinrichsLampartValentinMartin.2024,AlvarezLillLonigroValentinMartin.2025}. Especially, this explicit form of $H_{\sfr\sfe\sfn}(v)$ stems from \cite{HinrichsLampartValentinMartin.2024}. In those articles, the strong/norm resolvent continuity of $H_{\sfr\sfe\sfn}(v)$ as a function of $v$ is also proved. We will obtain the same result from our functional integration representation below.
\end{rem}
\subsection{Existence and Uniqueness of Ground States}
\label{subsec:GS}
Whereas the main result of this article is a Feynman--Kac representation of $H_{\sfr\sfe\sfn}(v)$ for any $v\in \fV_{\sfr\sfe\sfg}+\fV_{\sfU\sfV}$,
which we will derive in \cref{sec:FK}, one important application is the following result on the spectrum.
\begin{thm}
    \label{thm:GS}
    If $v\in\fV_{\sfU\sfV}$, then $H_{\sfr\sfe\sfn}(v)$ has a unique ground state, i.e.,
    \[\dim\ker\Big(H_{\sfr\sfe\sfn}(v)-\inf\sigma\big(H_{\sfr\sfe\sfn}(v)\big)\Big)=1.\]
\end{thm}
\begin{proof}
    We prove the claim in the end of \cref{sec:app}.
\end{proof}
\begin{rem}[Persistence of ground states]
    Under the additional assumption $v\in L^2(\cM,\mu)$, i.e., in the presence of an ultraviolet cutoff,
    the result is well-known; see for example \cite{BachFroehlichSigal.1997,Gerard.2000,HirokawaHiroshimaLorinczi.2014}.
    This is due to the fact, that the assumption $v/\omega \in L^2(\cM,\mu)$ in fact requires infrared regularity of the model.
    Our result thus proves the fact that existence of ground states is indeed an infrared property,
    since it is independent of the ultraviolet cutoff. We emphasize that both the uniqueness and the existence of the ground state
    do not follow from previous constructions of $H_{\sfr\sfe\sfn}(v)$.
\end{rem}
\begin{rem}[The infrared-critical spin boson model]
    If $v\in\fV_{\sfr\sfe\sfg}\setminus\fV_{\sfU\sfV}$,
    then it has been shown that the existence of a ground state
    undergoes a phase transition w.r.t. the coupling strength, i.e.,
    $H_{\sfr\sfe\sfn}(\alpha v)$ has a unique ground state for small $|\alpha|$,
    but none for $|\alpha|$ sufficiently large;
    cf. \cite{Spohn.1989,HaslerHinrichsSiebert.2021a,BetzHinrichsKraftPolzer.2025}.
    We \textbf{conjecture} that this behavior persists, when the ultraviolet cutoff
    is removed. Since the access to ground state existence for the critical ultraviolet regular model
    requires a Feynman--Kac based approach similar to \cref{sec:app}, see \cite{HaslerHinrichsSiebert.2021b,BetzHinrichsKraftPolzer.2025,HinrichsPolzer.2025},
    the construction of the path integral representation without a cutoff in this article
    is an important first step towards the proof of this conjecture.
\end{rem}

\section{Feynman--Kac Formula for the Spin Boson Semigroup}
\label{sec:FK}
In this \lcnamecref{sec:FK}, we derive a Feynman--Kac formula (FKF) for $H_{\sfr\sfe\sfn}(v)$ for any interaction $v\in\fV_{\sfr\sfe\sfg} + \fV_{\sfU\sfV}$.

\subsection{The Spin Process}
Starting point of our investigation is the probabilistic representation of the graph Laplacian on $\{\pm1\}$.
Therefore, let $\Omega$ be the space of right-continuous functions $\gamma:[0,\infty)\to\{\pm1\}$
and equip it with the $\sigma$-algebra $\Sigma$ generated by the pointwise evaluations $t\mapsto \gamma(t),\,t\ge 0$.
Then, let $(X_t)_{t\ge0}$ be the $\{\pm1\}$-valued continuous-time Markov chain on the filtered probability spaces
$(\Omega,\Sigma,\Sigma_t,\PP^x),\,x=\pm1$ with transition function
\begin{align}
    \label{eq:transition}
    p_t(x,y) = \frac12\big(1+e^{-2t}\delta_{x,y} - e^{-2t}\delta_{x,-y}\big),
\end{align}
that is
\begin{align}
    \label{eq:transition2}
    \PP^x[X_{t_1}=x_1,\ldots,X_{t_n}=x_n] = p_{t_1}(x,x_1)p_{t_2-t_1}(x_1,x_2)\cdots p_{t_n-t_{n-1}}(x_{n_1},x_n), \quad 0\le t_1\le\cdots\le x_n
\end{align}
and with the Markov property
\begin{align}
    \label{eq:Markov}
    \EE^x[X_{t+s}|\Sigma_s] = \EE^x[(X\circ \tau_s)_t| \Sigma_s] = \EE^{X_s}[X_t], \qquad\mbox{$\PP^x$-almost surely,}
\end{align}
where we introduced the notation that $\EE^x[\,\cdot\,]$ denotes integrals w.r.t. $\PP^x$ and $\tau_s:\Omega\to\Omega,\gamma\mapsto \gamma(\cdot+s)$,
cf. \cite[\S~2]{Liggett.2010} for more details on the general theory or \cite[\S~4.2]{Hinrichs.2022} for an explicit treatment of this specific process.
\begin{rem}
    In the literature, e.g., \cite{Spohn.1989,Abdessalam.2011,HirokawaHiroshimaLorinczi.2014,HaslerHinrichsSiebert.2021b}, an alternative definition of $X_t$ is used: denoting by $(N_t)_{t\ge 0}$ the Poisson process with intensity 1, one defines $S_t=(-1)^{N_t}$. Our process $X_t$ is then given as the pushforward measure on the path space of $(S_t)_{t\ge0}$.
\end{rem}
The following statement illustrates the generator property of the graph Laplacian for $X_s$.
Therein an from now on, we use Riemann--Stieltjes integrals when integrating w.r.t. functions of $X_t$.
In this way, we provide a fully analytic approach to the definition of integrals used throughout this article.
We do emphasize, however, that alternatively one could refer to stochastic integration w.r.t. the Poisson process.
Note that these are well-defined for any continuous integrand.
\begin{lem}
    \label{lem:stieltjesexp}
    Let $(A_t)_{t\ge 0}$ be a Banach-space valued adapted stochastic process  on $(\Omega,\Sigma,\Sigma_t)$ which has continuous paths $\PP^x$-almost surely and is uniformly bounded in $\Omega$. Then for any $f:\{\pm1\}\to\IC$ and $x\in\{\pm1\}$
    \begin{align*}
        \EE^x\Bigg[\int_0^t A_s \sfd f(X_s)\Bigg] = -\int_0^t \EE^x\big[A_s\big(f(X_s)-f(-X_s\big)\big]\sfd s.
    \end{align*}
\end{lem}
\begin{rem}
    \label{rem:stochint}
    In the language of stochastic differential equation, the above statement could also be inferred from
    \[
        \sfd f(X_s) = \big(f(X_{s_-})-f(-X_{s_-})\big) \sfd N_s,
    \]
    combined with It\^o's product formula.
    We here choose to give a selfcontained proof, solely on the basis of Riemann--Stieltjes integration.
\end{rem}
\begin{proof}
    Let $t_{i,n}\coloneqq t\cdot\frac in$. Then, for all $\gamma\in\Omega$, by definition
    \begin{align*}
        \int_0^t A_s(\gamma) \sfd f(\gamma_s) = \lim_{n\to\infty}\sum_{i=1}^n A_{t_{i-1}}(\gamma) \big(f(\gamma_{t_i}-f(\gamma_{t_{i-1}}\big).
    \end{align*}
    By the uniform boundedness of $A_s$ and the dominated convergence theorem,
    we can exchange integration w.r.t. $\PP^x$ and the limit $n\to\infty$ and obtain
    \begin{align*}
        \EE^x\Bigg[\int_0^t A_s \sfd f(X_s)\Bigg] &= \lim_{n\to\infty}\sum_{i=1}^n\EE^x\Big[ A_{t_{i-1}}(\gamma) \big(f(X_{t_i}-f(X_{t_{i-1}}\big)\Big]\\
            & = \lim_{n\to\infty}\sum_{i=1}^n\EE^x\Big[\EE^x\Big[ A_{t_{i-1}}(\gamma) \big(f(X_{t_i})-f(X_{t_{i-1}})\big| \Sigma_{t_{i-1}}\big)\Big]\Big]\\
            & = \lim_{n\to\infty}\sum_{i=1}^n\EE^x\Big[A_{t_{i-1}}\big(f(-X_{t_{i-1}})-f(X_{t_{i-1}})\big)\EE^x\Big[ \chr_{\{X_{t_{i-1}}\ne X_{t_i}\}}\big| \Sigma_{t_{i-1}}\Big]\Big],
    \end{align*}
    where we used the law of total expectation and the adaptedness of $A_t$ in the last two steps, respectively.
    Now, by the Markov property \cref{eq:Markov} and the concrete transition function of $X_t$ \cref{eq:transition},
    we have
    \begin{align*}
        \EE^x\Big[ \chr_{\{X_{t_{i-1}}\ne X_{t_i}\}}\big| \Sigma_{t_{i-1}}\Big] = \frac12\big(1-e^{-2t/n}\big) = -\frac tn + O(n^{-2}).
    \end{align*}
    Combining the last two observations and recognizing the convergent Riemann sum appearing therein,
    we find
    \begin{align*}
        \EE^x\Bigg[\int_0^t A_s \sfd f(X_s)\Bigg] &= \lim_{n\to\infty}\frac tn\sum_{i=1}^n\EE^x\Big[A_{t_{i-1}}\big(f(-X_{t_{i-1}})-f(X_{t_{i-1}})\Big] + O(n^{-1})
        \\& = \int_0^t \EE^x\Big[A_{s}\big(f(-X_{s})-f(X_{s})\Big]\sfd s.\qedhere
    \end{align*}
\end{proof}
\subsection{Hilbert Space-Valued Processes}
A major role in the FKF is played by the function-valued stochastic processes given for any path $\gamma\in\Omega$ by
\begin{align}
    \label{def:Ut}
	&U_t^+(v)(\gamma)(k) = \int_0^t \e^{-s\omega(k)} \gamma_sv(k)\d s, && U_t^-(v)(\gamma)(k) = \int_0^t \e^{-(t-s)\omega(k)} \gamma_s v(k) \d s,
\end{align}
which are well-defined for almost every $k\in\cM$, by the fact that the integrands are continuous except at finitely many points in the interval $[0,t]$.

We collect relevant properties of $U_t^\pm(v)(\gamma)$ below.
\begin{prop}\label{Properties U plus minus}
    Let $v=v_{\sfr\sfe\sfg}+v_{\sfU\sfV}$ with $v_{\sfr\sfe\sfg}\in\fV_{\sfr\sfe\sfg}$, $v_{\sfU\sfV}\in\fV_{\sfU\sfV}$ and let $\gamma\in\Omega$.
	Then:
    \begin{enumprop}
        \item\label{prop:U.1} $U_t^\pm(v)(\gamma)\in \fV_{\sfr\sfe\sfg} = \sD(\omega^{-1/2})$ and
            \[ \norm{U_t^\pm(v)(\gamma) / \min(1,\sqrt{t\omega}) } \le t\norm{v_{\sfr\sfe\sfg}/ \min(1,\sqrt{t\omega})} + \norm{v_{\sfU\sfV}/\omega}, \quad t>0;  \]
        \item\label{prop:U.2} we have the flow equations
            \begin{align*}
                U_{t+s}^\pm(v) = e^{\min(0,\pm s\omega)}U_t^\pm(v) + e^{\min(0,\mp t\omega)}U_s^\pm(v)\circ\tau_t;
            \end{align*}
        \item\label{prop:U.3} for any $g\in\sD(\omega)$, the maps $t\mapsto \braket{g,U_t^\pm(v)(\gamma)}$ are absolutely continuous for $t$ in a compact interval and differentiable at all continuity points of $\gamma$ with derivative
            \begin{align*}
                &\partial_t\!\Braket{g,U_t^+(v)(\gamma)} = \gamma_t \braket{g,e^{-t\omega}v_{\sfr\sfe\sfg}} + \gamma_t\braket{\omega g, e^{-t\omega} v_{\sfU\sfV}/\omega},\\
                &\partial_t\!\Braket{g,U_t^-(v)(\gamma)} = -\braket{\omega g, U_t^-(v)(\gamma)} + \gamma_t \braket{g,v_{\sfr\sfe\sfg}} + \gamma_t\braket{\omega g, v_{\sfU\sfV}/\omega}
            \end{align*}
        \item\label{prop:U.4}  $U_t^\pm(v)\circ X:\Omega\to L^2(\cM,\mu)$ are $\Sigma_t$-$\fB(L^2(\cM,\mu))$-measurable stochastic process.
    \end{enumprop}
\end{prop}
\begin{proof}
    The claim \subcref{prop:U.1} follows from the linearity of $v\mapsto U_t^\pm(v)(\gamma)$,
    the triangle inequality, the obvious estimate on $U_t^\pm(v_{\sfr\sfe\sfg})(\gamma)$ and by applying Fubini's theorem to estimate
    \begin{align}\nonumber
        \int_{\cM} \frac{|U^\pm_{t}(v_{\sfU\sfV})(\gamma)|^2}{\min(1,t\omega)} \sfd\mu \le \int_\cM \Bigg(\int_0^te^{-s\omega}|v_{\sfU\sfV}|\sfd s\Bigg)^2\frac{\sfd\mu}{\min(1,t\omega)} 
        = \int_\cM \frac{1-e^{-t\omega}}{\min(1,t\omega)}\cdot \frac{|v_{\sfU\sfV}|^2}{\omega^2}\sfd\mu \le \Big\|\frac {v_{\sfU\sfV}}{\omega}\Big\|^2.
    \end{align}
    The flow equation \subcref{prop:U.2} is straightforward from the definition \cref{def:Ut}.
    Also, to prove \subcref{prop:U.3},note that differentiability at continuity points of $\gamma$ follows from Fubini's theorem,
    the definition \cref{def:Ut} and the dominated convergence theorem.
    Thus, the absolute continity follows, by observing integrability of the derivative and continuity of $t\mapsto \braket{g,U_t^\pm(v)(\gamma)}$.
    Finally, the measurability \subcref{prop:U.4} of $U_t^\pm(v)\circ X$ follows from the fact that the Riemann integral (which is identical to the Lebesgue integral in this case)
    is the limit of $\Sigma_t--\fB(L^2(\mathcal{M}))$ measurable functions and therefore measurable itself.
\end{proof}
We will frequently drop the argument $X$, when considering $U_t^\pm\circ X:\Omega\to L^2(\cM,\mu)$ as stochastic process.
\subsection{The Phase Process}
Apart from the contributions given by $U^\pm_t(v)$, the vacuum expectation of the spin boson model is given by a scalar-valued stochastic process.
In the Schr\"odinger particle case, it is given by a self-attracting path measure \cite{Feynman.1955} for ultraviolet regular models
and by a Brownian stochastic integral for ultraviolet singular models \cite{MatteMoller.2018}.
Here, we provide an expression given by Riemann--Stieltjes integrals over the phase integrator
\begin{align}
    \nu(v)(t) \coloneqq -\int_{\cM}e^{-t\omega}\frac{|v|^2}{\omega}\sfd\mu,
\end{align}
which is easily observed to be well-defined for all $t>0$ and $v\in\fV_{\sfr\sfe\sfg}+\fV_{\sfU\sfV}$.
Let us collect important properties of $\nu$.
\begin{prop}
    \label{lem:nu}
    Let $v,f\in\fV_{\sfr\sfe\sfg} + \fV_{\sfU\sfV}$. Then:
    \begin{enumprop}
        \item\label{lem:nu.1} $\nu(v)\in L^1_{\sfl\sfo\sfc}([0,\infty))\cap C^1((0,\infty))$ with $\nu(v)(t)' = \int_{\cM}e^{-t\omega}|v|^2\sfd\mu$ for $t>0$;
        \item\label{lem:nu.2} if $\mu(\{vf\ne 0\})=0$, then $\nu(v+f)=\nu(v)+\nu(f)$;
        \item\label{lem:nu.3} if $v\in\fV_{\sfU\sfV}$, then
            \[ \int_0^t \nu(v)(t-s) \,\gamma_s\, \sfd s = -\braket{U_t^-(v)(\gamma),v/\omega}, \quad\ t>0,\ \gamma\in\Omega.\]
    \end{enumprop}
\end{prop}
\begin{proof}
    Claim \subcref{lem:nu.3} and \subcref{lem:nu.1} for $v\in\fV_{\sfU\sfV}$ follow from Fubini's theorem.
    The remaining part of \subcref{lem:nu.1} is immediate and follows by a direct calculation,
    noting the boundedness and continuity of the derivative by the assumptions.
    We stress that $\nu(\nu)(t)'$ diverges as $t\to0$, if $v\notin L^2(\cM,\mu)$.
    Claim \subcref{lem:nu.2} is immediate from the additivity of integrals.
\end{proof}
In view of \cref{lem:nu.1}, for any $v\in\fV_{\sfr\sfe\sfg}+\fV_{\sfU\sfV}$, we can now define the phase process $u_t(v):\Omega\to\IR$ as
\begin{align}
    \label{eq:defut}
    u_t(v)(\gamma) \coloneqq \int_0^t\Big( \gamma_0 \nu(v)(s) + \int_0^s \nu(v)(s-r)\sfd\gamma_r\Big)\gamma_s\sfd s.
\end{align}
Let us again collect properties of $u_t$.
\begin{prop}
    \label{prop:phase}
    Let $v=v_{\sfr\sfe\sfg}+v_{\sfU\sfV}$ with $v_{\sfr\sfe\sfg}\in\fV_{\sfr\sfe\sfg}$, $v_{\sfU\sfV}\in\fV_{\sfU\sfV}$ and $\mu(\{v_{\sfr\sfe\sfg}v_{\sfU\sfV}\ne0\})=0$ and let $\gamma\in\Omega$.
    Then:
    \begin{enumprop}
        \item\label{prop:phase.1} $\displaystyle u_t(v)(\gamma)=u_t(v_{\sfr\sfe\sfg})(\gamma) + u_t(v_{\sfU\sfV})(\gamma)$;\\[-.5em]
        \item\label{prop:phase.conv} if $(v_n)\subset\fV_{\sfr\sfe\sfg}+\fV_{\sfU\sfV}$ satisfies $v_n/(\sqrt{\omega}+\omega)\xrightarrow[n\to\infty]{L^2(\cM,\mu)} v/(\sqrt{\omega}+\omega)$, then $u_t(v_n)(\gamma)\xrightarrow{n\to\infty} u_t(v)(\gamma)$;
        \item\label{prop:phase.2}  $\displaystyle u_t(v_{\sfr\sfe\sfg})(\gamma) + t \norm{\omega^{-1/2}v_{\sfr\sfe\sfg}}^2 = \int_0^t \gamma_s\braket{U_s^-|v_{\sfr\sfe\sfg}}\d s= \int_0^t\int_0^s \braket{e^{-(s-r)\omega}v_{\sfr\sfe\sfg},v_{\sfr\sfe\sfg}}\gamma_r\gamma_s \sfd r\sfd s $;
        \item\label{prop:phase.3} $\displaystyle u_t(v_{\sfU\sfV})(\gamma) = -\gamma_t\braket{U_t^-(v_{\sfU\sfV})(\gamma)|v_{\sfU\sfV}/\omega} + \int_0^t\braket{U_{s}^-(v_{\sfU\sfV})(\gamma)|v_{\sfU\sfV}/\omega}\sfd\gamma_s$;\\[-.5em]
        \item\label{prop:phase.4} we have the flow equation
            $\displaystyle u_{t+s}(v)=u_t(v)+u_s(v) \circ \tau_t + \braket{ U_t^-(v) | U_s^+(v)\circ \tau_t }$;\\[-.5em]
        \item\label{prop:phase.5} $u_t(v)\circ X:\Omega\to \IR$ is $\Sigma_t$-$\fB(\IR)$-measurable.
    \end{enumprop}
\end{prop}
\begin{rem}
    Since the phase expression is most important in the analysis of the ground state regime,
    a few comments are in order:
    \begin{enumrem}
        \item\label{rem:phase.1} Comparing \subcref{prop:phase.2} and \cref{eq:utreg} with the choice
        \[W(r)=\int_{\cM}e^{-|r|\omega}|v_{\sfr\sfe\sfg}|^2\sfd\mu,\]
        one observes that this is the phase expression known from the literature
        \cite{FannesNachtergaele.1988,SpohnDuemcke.1985,Spohn.1989,Abdessalam.2011,HirokawaHiroshimaLorinczi.2014}.
        In this prospect, the definition \cref{eq:defut} extends the standard definition and incorporates more singular interactions.
        \item\label{rem:phase.2} Furthermore, in view of \cref{rem:stochint}, the Riemann--Stieltjes integrals in \subcref{prop:phase.3}
            relate to the stochastic integral expressions used for the (complex) phase
            in the Nelson model \cite{MatteMoller.2018}, the relativistic Nelson model
            \cite{HinrichsMatte.2022} or the Fr\"ohlich model \cite{HinrichsMatte.2024}.
            Thus, the simplicity of the spin boson model allows to unify the infrared and the ultraviolet part
            in one expression \cref{eq:defut}, a feature which is not available for the more advanced models
            in which the particle is a Schr\"odinger particle.
        \item\label{rem:phase.3}We also note that for $v\in\fV_{\sfr\sfe\sfg}\cap \fV_{\sfU\sfV}$ the right hand sides of \subcref{prop:phase.2} and \subcref{prop:phase.3} equal, by an application of the partial integration formula for the Riemann--Stieltjes integral, which is a simplification of the application of It\^o's formula in the referenced articles; see below proof for details.
    \end{enumrem}
\end{rem}
\begin{proof}
    Claim \subcref{prop:phase.1} is apparent from the definition \cref{eq:defut}, the assumption $\mu(\{v_{\sfr\sfe\sfg}\ne 0\ne v_{\sfU\sfV}\})=0$ and the additivity of integrals
    and \subcref{prop:phase.conv} follows by a simple application of the dominated convergence theorem.
    Further, \subcref{prop:phase.5} follows along the lines of \cref{prop:U.4}, by interpreting the defining integrals in the Riemannian sense.

    To prove \subcref{prop:phase.2}, first note that $\nu(v_{\sfr\sfe\sfg})\in C^1([0,\infty))$ and $\nu(v_{\sfr\sfe\sfg})(0)=-\norm{\omega^{-1/2}v}^2$. Now, denote by $0=T_0<T_1<\cdots< T_n\le T_{n+1}=t$ the jumps of $\gamma$ in the interval $[0,t]$. Then, by \cref{lem:nu.1} and using $\gamma_{T_i}=\gamma_0(-1)^i = \gamma_s$ for $s\in[T_i,T_{i+1})$, the right hand side equals
    \begin{align*}
        &\int_0^t\int_0^s \braket{e^{-(s-r)\omega}v_{\sfr\sfe\sfg},v_{\sfr\sfe\sfg}}\gamma_r\gamma_s \sfd r\sfd s
        \\
        &=\sum_{i=0}^n\int_{T_i}^{T_{i+1}}\Bigg(\sum_{j=1}^i\gamma_{T_i}\gamma_{T_{j-1}}\int_{T_{j-1}}^{T_{j}} \nu'(v_{\sfr\sfe\sfg})(s-r) \sfd r\sfd s
          +  \int_{T_i}^{s} \nu'(v_{\sfr\sfe\sfg})'(s-r) \sfd r\sfd s\Bigg)\\
        & = \sum_{i=0}^n\int_{T_i}^{T_{i+1}} \Bigg(\sum_{j=1}^i (-1)^{i+j}\Big(\nu(v_{\sfr\sfe\sfg})(s-T_{j}) - \nu(v_{\sfr\sfe\sfg})(s-{T_{j-1}})\Big)
             +\nu(v_{\sfr\sfe\sfg})(s-T_i) -  \nu(v_{\sfr\sfe\sfg})(0) \Bigg)\sfd s
        \\
        & = \int_0^t \gamma_0\gamma_s \nu(v_{\sfr\sfe\sfg})(s) \sfd s + \int_0^t \sum_{j=1}^n \chr_{\{T_j<s\}}2\gamma_0\gamma_s(-1)^j \nu(v_{\sfr\sfe\sfg})(s-T_j) \sfd s - t \nu(v_{\sfr\sfe\sfg})(0).
    \end{align*}
    Thus, \subcref{prop:phase.2} now follows from the observation that for any continuous function $f:[0,\infty)\to\IC$ the Riemann--Stieltjes integral w.r.t. $\gamma_r$ satisfies
    \begin{align*}
        \int_0^s f(r)\sfd\gamma_r = \sum_{j=1}^n \chr_{\{T_j<s\}}2\gamma_0(-1)^jf(T_j).
    \end{align*}
    The claim \subcref{prop:phase.3}, at least if $v_{\sfU\sfV}\in L^2(\cM,\mu)$,
    can be inferred similar to above, by integrating out $s$ first in above argument.
    Alternatively, along the lines of \cref{rem:phase.3}, we apply the partial integration formula
    for Riemann--Stieltjes integrals to $\gamma_t\braket{U_t^-(v)(\gamma),v_{\sfU\sfV/\omega}}$ and using \cref{prop:U.3} obtain
    \begin{align*}
        \gamma_t\braket{U_t^-(v)(\gamma),v_{\sfU\sfV}/\omega}
        &= \int_0^t \braket{U_s^-(v)(\gamma),v_{\sfU\sfV}/\omega}\sfd\gamma_s - \int_0^t \braket{U_s^-(v)(\gamma),v_{\sfU\sfV}}\sfd\gamma_s + t\norm{\omega^{-1/2}v_{\sfU\sfV}}^2,
    \end{align*}
    whence \subcref{prop:phase.2} implies \subcref{prop:phase.3}, by simple rearrangement. The general claim then follows, by \subcref{prop:phase.conv} and a straightforward approximation procedure.

    To prove the flow equation \subcref{prop:phase.4}, we first observe that \subcref{prop:phase.2} and \cref{prop:U.2} imply
    \begin{align*}
        u_t(v_{\sfr\sfe\sfg})
        & = \int_0^{t+s}\gamma_r\braket{U_r^-(v_{\sfr\sfe\sfg}),v_{\sfr\sfe\sfg}} \sfd r + (t+s)\norm{\omega^{-1/2}v_{\sfr\sfe\sfg}}^2
        \\& = u_t(v_{\sfr\sfe\sfg}) + \int_0^{s}\gamma_{t+r}\braket{U_{t+r}^-(v_{\sfr\sfe\sfg}),v_{\sfr\sfe\sfg}} \sfd r + s\norm{\omega^{-1/2}v_{\sfr\sfe\sfg}}^2
        \\& = u_t(v_{\sfr\sfe\sfg}) + u_s(v_{\sfr\sfe\sfg})\circ\tau_t +  \int_0^s\braket{e^{-r\omega}U_r^-,\gamma_{t+r}v_{\sfr\sfe\sfg}}\sfd r.
    \end{align*}
    Now, by Fubini's theorem, the last summand is $\braket{ U_t^-(v_{\sfr\sfe\sfg}) | U_s^+(v_{\sfr\sfe\sfg})\circ \tau_t }$, which proves the claim for $v\in\fV_{\sfr\sfe\sfg}$. The general statement again follows from \subcref{prop:phase.conv} and \cref{prop:U.1} (combined with the linearity of $U_t^\pm$ in $v$), by approximation.
    
    This completes the proof.
\end{proof}
\subsection{Feynman--Kac Formula}
We can now state and prove our main result.
In the statement and throughout this section, we equip the interaction space with the norm
\begin{align*}
    \norm{v}_{\omega} \coloneqq \|(\sqrt{\omega}+\omega)^{-1}v\|^2_{L^2(\cM,\mu)}, \quad v\in\fV_{\sfr\sfe\sfg} + \fV_{\sfU\sfV}.
\end{align*}
\begin{thm}
    \label{thm:FKF}
    Let $v\in\fV_{\sfr\sfe\sfg}+\fV_{\sfU\sfV}$.
    Then
    \begin{align}
        \label{def:Tt}
        \big(T_t(v) \psi\big)(x) \coloneqq \EE^x\Big[ e^{u_t(v)(X)}I_{t/2}\big(\omega,U^+_t(v)(X)\big)\!^*I_{t/2}\big(\omega,U^-_t(v)(X)\big)\psi(X_t) \Big],\quad \psi\in L^2(\{\pm1\};\FS)
    \end{align}
    defines a strongly continuous semigroup on $L^2(\{\pm1\};\FS)$ with generator $H_{\sfr\sfe\sfn}(v)$, i.e.,
    \begin{align}
        e^{-t H_{\sfr\sfe\sfn}(v)} = T_t(v), \quad t>0.
    \end{align}
    Furthermore, the map $\fV_{\sfr\sfe\sfg}+\fV_{\sfU\sfV} \to \cB(L^2(\{\pm1\};\FS))$, $v\mapsto T_t(v)$ is continuous.
\end{thm}
We divide the proof into several lemmas, starting with an analysis of the $\cB(\FS)$-valued integrand
in the Feynman--Kac formula \cref{def:Tt}, i.e.,
\begin{align}
    \label{def:Wt}
	W_0(v)(\gamma)\coloneqq \Id_{\FS},\quad W_t(v)(\gamma) \coloneqq e^{u_t(v)(\gamma)}I_{t/2}\big(\omega,U^+_t(v)(\gamma)\big)^*I_{t/2}\big(\omega,U^-_t(v)(\gamma)\big),\quad t>0.
\end{align}
\begin{prop}
	\label{prop:Wt}
    Let $v=v_{\sfr\sfe\sfg}+v_{\sfU\sfV}$ with $v_{\sfr\sfe\sfg}\in\fV_{\sfr\sfe\sfg}$, $v_{\sfU\sfV}\in\fV_{\sfU\sfV}$ and $\mu(\{v_{\sfr\sfe\sfg}v_{\sfU\sfV}\ne0\})=0$ and let $\gamma\in\Omega$.
    Then, for all $t\ge 0$, we have that
    \begin{enumprop}
        \item\label{prop:Wt.1} the map $W_t(v)\circ X:\Omega\to\cB(\FS)$ is $\Sigma_t$--$\mathfrak{B}(\cB(\mathcal{F}))$-measurable;
        \item\label{prop:Wt.bound} $\gamma\mapsto \sup_{s\in[0,t]}\sup_{\|v\|_{\omega}\le 1}\|W_s(v)(\gamma)\|_{\cB(\FS)}\in L^1(\Omega,\PP^x)$ for $x=\pm1$;
        \item\label{prop:Wt.cont2} the map $\fV_{\sfr\sfe\sfg} + \fV_{\sfU\sfV}\to \cB(\FS)$, $v\mapsto W_t(v)(\gamma)$ is continuous w.r.t. the norm operator topology;
        \item\label{prop:Wt.cont} the map $[0,\infty)\to \cB(\FS)$, $t\mapsto W_t(v)(\gamma)$ is continuous w.r.t. the strong operator topology;
        \item\label{prop:Wt.2} $W_{t+s}(v)=W_t(v)(W_s(v)\circ\tau_t)$ for all $s\ge 0$;
        \item\label{prop:Wt.3} for all  $\xi\in\cE_{\sfr\sfe\sfn}(+1,v)+\cE_{\sfr\sfe\sfn}(-1,v)+\cE(\sD(\omega))$ and $\zeta\in\cE(L^2(\cM,\mu))$,
            the map $t\mapsto \braket{\zeta,W_t(v)(\gamma)\xi}$ is absolutely continuous
             and if $t$ is a continuity point of $\gamma$ such that $\xi\in\cE_{\sfr\sfe\sfn}(\gamma_t,v)$, then
             \[
                \partial_t \braket{\zeta,W_t(v)(\gamma)\xi}
                = -\braket{\zeta,W_t(v)(\gamma)h(v,\gamma_t)\xi}
                .
            \]
    \end{enumprop}
\end{prop}
\begin{rem}
    In the literature \cite{GueneysuMatteMoller.2017,MatteMoller.2018,HinrichsMatte.2022},
    similar differential equations to \subcref{prop:Wt.3} are at the core of the proof of the corresponding FKFs.
    In all the referenced articles, however, it can only be derived in the presence of a suitable ultraviolet regularization,
    which would correspond to the assumption $v\in\fV_{\sfr\sfe\sfg}$.
    We thus stress that our explicit construction of $h(v,x)$ \cref{eq:vHren}
    allows for a direct identification of the self-energy renormalized operator.
\end{rem}
\begin{proof}
    As shown in \cref{prop:U.4,prop:phase.5}, respectively, $U_{t}^{\pm}(v)\circ X$ and $u_{t}(v)\circ X$ are measurable w.r.t. $\Sigma_t$.
    Thus, in view of \cref{prop:U.1}, \cref{eq:Ftbound} and the analyticity of $h\mapsto I_t(\omega,h)$ for $h\in\fV_{\sfr\sfe\sfg}$,  $W_t(v)\circ X$ is a composition of measurable maps, which proves \subcref{prop:Wt.1}.
    Furthermore, \subcref{prop:Wt.bound} and \subcref{prop:Wt.cont2} easily follow from \cref{prop:U.1} (and the linearity of $U_t(v)(\gamma)$ in $v$) as well as \cref{eq:Ftbound,eq:defut}.
    
    For the proof of the remaining statements, first observe that \cref{eq:It}, \cref{eq:Itstar} and the definition \cref{def:Wt} imply
    \begin{align*}
         W_t(v)(\gamma)\eps(h) = e^{u_t-\braket{U_t^-(v)(\gamma),h}}\eps(e^{-t\omega}h-U_t^+(v)(\gamma)),
         \quad h\in L^2(\cM,\mu).
    \end{align*}
    Thus, by \cref{eq:Ftbound,prop:U.1,prop:U.2,prop:phase.4} as well as the density of $\cE(L^2(\cM,\mu))$ in $\FS$, \subcref{prop:Wt.cont} and \subcref{prop:Wt.2} follow immediately.
    Further, taking the scalar product with another exponential vector and using $\braket{\eps(f),\eps(g)}=e^{\braket{f,g}}$, we find
    \begin{align*}
        &
        \braket{\eps(g),W_t(v)(\gamma)\eps(h)} = \exp\big({u_t(v) - \braket{g,U_t^+(v)(\gamma)} + \braket{e^{-t\omega} g - U_t^-(v)(\gamma),h}}\big),
    \end{align*}
    whence the claimed absolute continuity in \subcref{prop:Wt.3} follows from \cref{prop:U.3,lem:nu.3,eq:defut}.
    Furthermore, if $h\in\cE_{\sfr\sfe\sfn}(\gamma_t,v)$, i.e., $\tilde h \coloneqq h+ \gamma_t v_{\sfU\sfV}/\omega \in \sD(\omega)$, by \cref{prop:U.3,prop:phase}, we can differentiate and obtain
    \begin{align*}
        \partial_t\big( & u_t(v) - \braket{g,U_t^+(v)} + \braket{e^{-t\omega} g - U_t^-,h}\!\big)\\
        &= \partial_t\big(\!\braket{e^{-t\omega}g,h} + u_t(v_{\sfr\sfe\sfg})(\gamma) -\braket{g,U_t^+(v)(\gamma)} - \braket{U_t^-(v)(\gamma),\tilde h}\!\big)\\
        &= \braket{-e^{-t\omega}\omega g-\gamma_t v_{\sfr\sfe\sfg},\tilde h} - \gamma_t \braket{e^{-t\omega}g - U_t^-(v)(\gamma),v_{\sfr\sfe\sfg}} - \norm{\omega^{-1/2}v_{\sfr\sfe\sfg}}^2 + \braket{U_t^-(v)(\gamma)-\gamma_tv_{\sfU\sfV}/\omega,\omega\tilde h}.
    \end{align*}
    Directly from the definition \cref{eq:vHren}, arguing similar to the above derivation,
    employing the observation $\braket{\eps(f),\dG(\omega)\eps(\tilde h)} = \braket{f,\omega \tilde h}e^{\braket{f,g}}$ as well as \cref{def:ann},
    and denoting $\tilde g \coloneqq e^{-t\omega}g - U_t^-(v)(\gamma)+\gamma_tv_{\sfU\sfV}/\omega$,
    \begin{align*}
        \tilde u = -\norm{v_{\sfU\sfV}/\omega}^2 - \gamma_t\braket{{e^{-t\omega}g-U_t^{-}(v),v_{\sfU\sfV}/\omega}} - \gamma_t\braket{v_{\sfU\sfV}/\omega,h}+{u_t-\braket{g,U_t^+(v)}}
    \end{align*}
    we now find
    \begin{align*}
         &\braket{\eps(g),W_t(v)(\gamma)h(v,\gamma_t)\eps(h)}
         = e^{\tilde u}\Braket{\eps(\tilde g),\big(\dG(\omega)+\gamma_t\ph(v_{\sfr\sfe\sfg}) + \norm{\omega^{-1/2}v_{\sfr\sfe\sfg}}^2\big)\eps(\tilde h)}
         \\& \qquad = \partial_t\big(  u_t(v) - \braket{g,U_t^+(v)} + \braket{e^{-t\omega} g - U_t^-,h}\!\big) e^{\tilde u}\braket{\eps(\tilde g),\eps(\tilde h)}.
    \end{align*}
    Thus, \subcref{prop:Wt.3} follows from the simple observation
      \begin{align*}
        &e^{\tilde u}\braket{\eps(\tilde g),\eps(\tilde h)}
         = \braket{\sW(\gamma_t v_{\sfU\sfV}/\omega) W_t(v)(\gamma)^*\eps(g),\sW(\gamma_t v_{\sfU\sfV}/\omega)\eps(h)} = \braket{\eps(g),W_t(v)(\gamma)\eps(h)}.\qedhere
    \end{align*}
\end{proof}
This immediately entails that the right hand side of \cref{def:Tt} indeed defines a strongly continuous semigroup.
\begin{cor}
    For $v\in\fV_{\sfr\sfe\fg}+\fV_{\sfU\sfV}$, the map $[0,\infty)\to\cB(L^2(\{\pm1\};\FS))$, $t\mapsto T_t(v)$ given by \cref{def:Tt} defines a strongly continuous semigroup and the map $\fV_{\sfr\sfe\fg}+\fV_{\sfU\sfV}\to\cB(L^2(\{\pm1\};\FS))$, $v\mapsto T_t(v)$ is continuous w.r.t. operator norm.
\end{cor}
\begin{proof}
    We first observe that
    \begin{align*}
        T_t(v)\psi(x) = \EE^x\big[ W_t(v)(X)\psi(X_t)\big]
    \end{align*}
    indeed defines a bounded operator on $L^2(\{\pm1\};\FS)$,
    by \cref{prop:Wt.bound}.
    The semigroup properties follows from \cref{prop:Wt.2,prop:Wt.3}
    and the Markov property of $X$, by the direct observation
     \begin{align*}
     T_t(v)T_s(v)f(x) & = \int_{\Omega}\int_\Omega W_t(v)(\gamma)W_s(v)(\sigma)f(\sigma_s)\sfd\PP^{\gamma_t}(\sigma)\sfd\PP^x(\gamma)\\
     & = \EE^x\big[W_t(v)(X)\EE[ W_s(v)(\tau_t(X))f(X_{t+s} |\Sigma_t]\big]\\
     & = \EE^x\big[W_t(v)(X)W_s(v)(X_{t+s})f(X_{t+s})\big]\\
     & = T_{t+s}(v)f(x).
    \end{align*}
    The strong continuity follows from the fact that the paths of $\gamma$ are right-continuous,
    the continuity of $t\mapsto W_t(v)(\gamma)$ in the strong operator topology \cref{prop:Wt.cont}, the bound \cref{prop:Wt.bound} and the dominated convergence theorem.
    Finally, the norm continuity of $v\mapsto T_t(v)$ follows from \cref{prop:Wt.cont2,prop:Wt.bound} and the dominated convergence theorem.
\end{proof}
We move to the
\begin{proof}[\textbf{Proof of \cref{thm:FKF}}]
    It only remains to prove that $H_{\sfr\sfe\sfn}(v)$ is the generator of $T_t(v)$, which follows if we can prove
    \begin{align}
        \braket{\Phi,T_t(v)\Psi} - \braket{\Phi,\Psi} = -\int_0^t\braket{\Phi,T_s(v)H_{\sfr\sfe\sfn}(v)\Psi}\sfd s, \quad \Phi,\Psi\in \bigoplus_{x=\pm 1} \cE_{\sfr\sfe\fn}(x,v),
    \end{align}
    by recalling the the chosen test function space is a core for $H_{\sfr\sfe\sfn}(v)$.
    To this end, the absolute continuity \cref{prop:Wt.3}, we find 
    \begin{align*}
        &\braket{\Phi(x),W_t(v)(X)\Psi(X_t)}_{\FS} - \braket{\Phi(x),\Psi(X_0)}
        \\&= -\int_0^t \braket{\Phi(x),W_s(v)(X)h(v,\gamma_s)\Psi(X_s)}\sfd s + \int_0^t \braket{\Phi(x),W_s(v)(X)\sfd \Psi(X_s)},
    \end{align*}
    where the last term on the right hand side can be interpreted as a Riemann--Stieltjes integral or (equivalently) a sum over the jumps of $X_s$.
    Taking the expectation w.r.t. $\PP^x$ and applying \cref{lem:stieltjesexp}, we find
    \begin{align*}
        \EE^x[\braket{\Phi(x),W_t(v)(X)\Psi(X_t)}_{\FS}] - \braket{\Phi(x),\Psi(x)} = -\int_0^t \EE^x[\braket{\Phi(x),W_s(v)\big(H_{\sfr\sfe\sfn}(v)\Psi\big)(X_s)}]\sfd s
    \end{align*}
    Thus, summing over $x=\pm 1$ and comparing with the definition \cref{def:Tt} proves the statement.
\end{proof}
 
\section{Application: Ground State Properties}
\label{sec:app}
We now apply \cref{thm:FKF}, to study the ground state regime of $H_{\sfr\sfe\sfn}(v)$.

\subsection{Ergodicity in the \texorpdfstring{$\sQ$}{Q}-space representation}
The $\sQ$-space is a natural $L^2$-representation of the Fock space $\FS$, i.e.,
$(\sQ,\fq)$ is a probability space such that $\FS$ and $L^2(\sQ,\fq)$ are unitarily equivalent.
For $v\in\fV_{\sfr\sfe\sfg}$, it is well-known (and follows, e.g., from a simple perturbative argument)
that, if $(\sQ,\fq)$ is appropriately chosen, $e^{-t H_{\sfS\sfB}(v)}$ improves positivitiy (or equivalently, ergodic) on $L^2(\sQ,\fq)$, i.e.,
$e^{-t H_{\sfS\sfB}(v)}f$ is strictly positive almost everywhere for $t>0$ and $0\le f\in L^2(\sQ,\fq)\setminus\{0\}$.
However, since \cref{thm:SEren} involves a limiting process,
this property does not immediately carry over to $e^{-t H_{\sfr\sfe\sfn}(v)}$ for any $v\in\fV_{\sfr\sfe\sfn}+\fV_{\sfU\sfV}$.
We here demonstrate that it is a direct corollary of the Feynman--Kac formula \cref{thm:FKF}.

More detailed introductions of the $\sQ$-space can, e.g., be found in \cite{Simon.1974,HiroshimaLorinczi.2020}
or also see \cite[\S~5.1]{HinrichsPolzer.2025} for an explicit overview for the spin boson model.
For our purposes, it suffices to recall that we can choose the unitary $\Theta:\FS\to L^2(\cQ,\fq)$ such that
\begin{enumerate}
    \item the vacuum $\eps(0)$ maps to the constant function one, i.e., $\Theta\eps(0)=1$;
    \item $\Theta\ph(f)\Theta^*$ is a multiplication operator for any real-valued $f$;
    \item $\Theta e^{-t\dG(\omega)}\Theta^*$ improves positivity for any $t>0$.
\end{enumerate}
In \cite[Prop.~8.1]{MatteMoller.2018}, the authors proved that $\Theta I_t(\omega,f)\Theta^*$ and
$\Theta I_t(\omega,f)^*\Theta^*$ improve positivity for any real-valued $f\in L^2(\cM,\mu)$.
Thus, the following result immediately follows from the fact that $u_t(v)(\gamma)$ is real-valued,
by definition \cref{eq:defut} and that $U_t^\pm(v)(\gamma)$ is real-valued by definition \cref{def:Ut}, if $v$ is.
\begin{prop}
    \label{prop:erg}
    If $v\in\fV_{\sfr\sfe\sfg}+\fV_{\sfU\sfV}$ is real-valued, then $\Theta T_t(v)\Theta^*$ improves positivity on $L^2(\{\pm1\}\times\sQ,2^{-1/2}\delta\otimes\fq)$.
\end{prop}
\begin{proof}
    If $x=\pm 1$, $q\in\sQ$, then for any $\Psi\in L^2(\{\pm1\};\FS)$ \cref{thm:FKF} implies
    \begin{align*}
         \Big(\Theta T_t(v)\Psi(x)\Big)(q) = \EE^x\Big[e^{u_t(v)(X)}U\big(I_{t/2}(\omega,U_t^+(v)(X))^*I_{t/2}(\omega,U_t^-(v)(X))\Psi(X_t)\big)(q)\Big].
    \end{align*}
    Hence, if $\Theta\Psi$ is non-negative and non-zero, then the right hand side is an average over a non-zero non-negative function, by the above considerations.
    This proves the claim.
\end{proof}
\begin{rem}
    \label{rem:realv}
    The assumption that $v$ be real-valued is in fact only a technical restriction on the choice of positive cone,
    as has, e.g., been observed in \cite{Moller.2005,HinrichsPolzer.2025}. In fact, if $v$ is not real-valued,
    then under the unitary transformation $\Gamma(\chr_{v\ne 0}\overline{v}/|v| + \chr_{v=0})$
    (with the argument interpreted as multiplication operators), the operator $H_{\sfr\sfe\sfn}(v)$ transforms to
    $H_{\sfr\sfe\sfn}(|v|)$, whence all corollaries independent of the concrete positive cone remain valid.
\end{rem}
\begin{rem}
    As observed in the beginning of this section,
    proving that an operator obtained by a limiting procedure
    is itself ergodic is not straightforward and
    requires additional information on the structure of the operator itself.
    To our knowledge, in the case of the Nelson model, ergodicity for the
    model without cutoff albeit expected to hold was first proven in \cite{MatteMoller.2018},
    also employing a Feynman--Kac formula. Analogs for the two-dimensional
    relativistic Nelson model were derived in \cite{Sloan.1974b,HinrichsMatte.2022}.
    We also note that ergodicity w.r.t. the so-called Fr\"ohlich cone for translation-invariant models
    with fixed total momentum was expected independent of ultraviolet cutoffs \cite{Frohlich.1973,Frohlich.1974},
    but only in the recent years proved using various technical approaches \cite{Miyao.2019,Lampart.2020,HinrichsHiroshima.2024}.
\end{rem}
By the Perron--Frobenius--Faris theorem, cf. \cite{Faris.1972} or \cite[\S~XIII.12]{ReedSimon.1978},
we now have the following immediate consequence.
\begin{cor}
    \label{cor:PFF}
    If $H_{\sfr\sfe\sfn}(v)$ has a ground state $\psi$, i.e., $H_{\sfr\sfe\sfn}(v)\psi = E(v)\psi$ for $E(v)=\inf\sigma(H_{\sfr\sfe\sfn}(v))$, then it is non-degenerate, i.e., $\ker(H_{\sfr\sfe\sfn}(v)-E)=\Span\{\psi\}$. Furthermore, setting $\Omega_\downarrow(x)=\eps(0)$ for $x=\pm1$, we have $\braket{\psi,\Omega_\downarrow}\ne 0$.
\end{cor}
\begin{proof}
    If $v\in\fV_{\sfr\sfe\sfg}+\fV_{\sfU\sfV}$ is real-valued, this is a standard corollary of the Perron--Frobenius--Faris theorem, recalling that $\Theta\Omega_\downarrow(x,q)=1$ for $x=\pm 1$, $q\in\sQ$.
    In the general case, by \cref{rem:realv} and the fact that the unitary therein leaves $\Omega$ invariant, it then follows from the unitary equivalence of $H_{\sfr\sfe\sfn}(v)$ to $H_{\sfr\sfe\sfn}(|v|)$.
\end{proof}
\subsection{Existence of the Ground State}
\Cref{cor:PFF} from the previous section also allows us to study the existence of ground states,
by combining that ergodicity of $H_{\sfr\sfe\sfn}(v)$ and positivity of $\Omega_{\downarrow}$ yield
\begin{align}
    \label{eq:Wienertype}
    \|\chr_{E(v)}(H_{\sfr\sfe\sfn}(v))\Omega_{\downarrow}\|^2 = \lim_{t\to\infty} \frac{\braket{\Omega_\downarrow,e^{-tH_{\sfr\sfe\sfn}(v)}\Omega_{\downarrow}}^2}{\braket{\Omega_\downarrow,e^{-2tH_{\sfr\sfe\sfn}(v)}\Omega_{\downarrow}}},
\end{align}
see for example \cite{HirokawaHiroshimaLorinczi.2014,HiroshimaLorinczi.2020,HinrichsPolzer.2025},
with the Feynman--Kac formula \cref{thm:FKF} applied to $\Omega_\downarrow$ yielding
\begin{align}
    \label{eq:FKFvac}
    \braket{\Omega_\downarrow,e^{-tH_{\sfr\sfe\sfn}(v)}\Omega_{\downarrow}} = \frac12\sum_{x=\pm 1}\EE^x[e^{u_t(v)(X)}].
\end{align}
This allows for the following simple
\begin{proof}[\textbf{Proof of \cref{thm:GS}}]
    By the assumption $v\in\fV_{\sfU\sfV}$ and \cref{prop:U.1,prop:phase.4}, we have the simple estimate
    \begin{align*}
        u_{2t}(v)(\gamma) = u_t(v)(\gamma)+u_t(v)(\tau_t(\gamma)) + \braket{U_t^-(v)(\gamma),U_t^+(v)(\tau_t(\gamma))} \le u_t(v)(\gamma)+u_t(v)(\tau_t(\gamma)) + \norm{v/\omega}^2.
    \end{align*}
    Noting that $u_t$ is independent of the starting value of $\gamma$, i.e., $u_t(v)(\gamma)=u_t(v)(-\gamma)$,
    and once more employing the Markov property of $X$ as well as the tower property of conditional expectation and \cref{prop:phase.5}, we find for $x=\pm 1$
    \begin{align*}
        e^{-\|v/\omega\|^2}\EE^x[e^{u_{2t}(v)(X)}] &\le \EE^x[e^{u_t(v)(X)+u_t(v)(\tau_t(X))}] = \EE^x[\EE^x[e^{u_t(v)(X)+u_t(v)(\tau_t(X))}|\Sigma_t]] \\&\quad = \EE^x[e^{u_t(v)(X)}\EE^{X_t}[e^{u_t(v)(X)]}] = \EE^x[e^{u_t(v)(X)}]^2.
    \end{align*}
    Combining this observation with \cref{eq:Wienertype,eq:FKFvac}, we find
    \begin{align}
    \label{eq:GSoverlap}
        \|\chr_{E(v)}(H_{\sfr\sfe\sfn}(v))\Omega_{\downarrow}\|^2\ge e^{-\|v/\omega\|^2}
    \end{align}
    and thus by \cref{cor:PFF} $H_{\sfr\sfe\sfn}(v)$ has a unique ground state.
\end{proof}
\begin{rem}
    Proofs for the existence of ground states of the infrared regular spin boson model
    in presence of an ultraviolet cutoff were previously given in \cite{HirokawaHiroshimaLorinczi.2014,HinrichsPolzer.2025}.
    Especially, in the second referenced article,
    also the estimate \cref{eq:GSoverlap} for the ground state vacuum overlap is derived.
    Our derivation allows for a simple probabilistic treatment of the problem,
    reducing the required ingredients to the shift relation \cref{prop:phase.4}
    involving a uniformly bounded error term (coming from the infrared regularity),
    the Markov property of the underlying process, and the independence of $u_t(v)(\gamma)$
    of the starting point $\gamma(0)$. We thus assume that this simple method
    allows for applications to further models from non-relativistic quantum field theory.
\end{rem}
%


\bibliographystyle{halpha-abbrv}
\bibliography{../../Literature/00lit}

\end{document}